\documentclass[letterpaper, 11pt]{article}
\usepackage{fullpage}
\usepackage{graphicx}
\usepackage{amsmath, amsthm, amssymb, cite, color, comment, mathtools, url}
\usepackage[unicode]{hyperref}

\newtheorem{theorem}{Theorem}[section]
\newtheorem{lemma}[theorem]{Lemma}

\newtheorem{corollary}[theorem]{Corollary}

\theoremstyle{definition}

\newtheorem*{problem}{Problem}

\theoremstyle{remark}
\newtheorem{remark}[theorem]{Remark}

\title{An FPT Algorithm for the Exact Matching Problem and NP-hardness of Related Problems}
\author{
  Hitoshi Murakami\thanks{Independent Researcher (Graduated from Osaka University on March 2024), Japan.} \and
  Yutaro Yamaguchi\thanks{Osaka University, Japan. Email: \texttt{yutaro.yamaguchi@ist.osaka-u.ac.jp}}}
\date{\empty}

\begin{document}
\maketitle
\thispagestyle{empty}

\begin{abstract}
The exact matching problem is a constrained variant of the maximum matching problem: given a graph with each edge having a weight $0$ or $1$ and an integer $k$, the goal is to find a perfect matching of weight exactly $k$.
Mulmuley, Vazirani, and Vazirani (1987) proposed a randomized polynomial-time algorithm for this problem, and it is still open whether it can be derandomized.
Very recently, El Maalouly, Steiner, and Wulf (2023) showed that for bipartite graphs there exists a deterministic FPT algorithm parameterized by the (bipartite) independence number.
In this paper, by extending a part of their work, we propose a deterministic FPT algorithm in general parameterized by the minimum size of an odd cycle transversal in addition to the (bipartite) independence number.
We also consider a relaxed problem called the correct parity matching problem, and show that a slight generalization of an equivalent problem is NP-hard.
\end{abstract}


\clearpage
\setcounter{page}{1}

\section{Introduction}

The \emph{maximum matching problem} is a fundamental problem in combinatorial optimization: given a graph, the goal is to find a matching of maximum size.
This problem is well-known to be polynomial-time solvable by deterministic algorithms, initiated by Edmonds \cite{Edmonds1965}.

The \emph{exact matching problem (EM)} is a constrained variant.
We say that a graph is \emph{0/1-weighted} if each edge has a weight $0$ or $1$, and define the weight of an edge subset (a matching, a path, a cycle, etc.) as the sum of the weights of edges in it.
The problem is stated as follows.

\begin{problem}[Exact Matching (EM)]
\mbox{ }
\begin{description}
\setlength{\itemsep}{0mm}
\item[Input:] A 0/1-weighted graph $G$ and an integer $k$.
\item[Task:] Determine whether $G$ has a perfect matching of weight $k$.
\end{description}
\end{problem}

This problem was introduced by Papadimitriou and Yannakakis~\cite{PY1982}, and is known to be polynomial-time solvable by a \emph{randomized}  algorithm proposed by Mulmuley, Vazirani, and Vazirani~\cite{MVV1987}.
It is, however, still widely open whether there exists a \emph{deterministic} polynomial-time algorithm or not, although several special cases have been solved, e.g., complete or comlete bipartite graphs~\cite{Karzanov1987, YMS2002, GS2011}, $K_{3,3}$-minor-free graphs~\cite{Yuster2012}, and bounded-genus graphs~\cite{GL1999}.
There are also several research directions such as relaxation, approximation, and clarifying relations with other problems~\cite{DEW2023, ElMaalouly2023, ES2022, ESW2023, Yuster2012}; see \cite{EHW2024} for more details on its history.

\subsection{FPT Algorithms for Exact Matching Problem}
For a parameter $k$ related to the input, a problem is said to be \emph{fixed-parameter tractable (FPT)} if there exists and a computable function $f \colon \mathbb{Z} \to \mathbb{Z}$ such that the problem can be solved in $f(k) \cdot n^{O(1)}$ time, where $n$ is the input size.
We also call an algorithm \emph{FPT (parameterized by $k$)} if it admits such a computational time bound.

Very recently, El Maalouly, Steiner, and Wulf~\cite{ESW2023} proposed a deterministic FPT algorithm for EM in bipartite graphs parameterized by the bipartite independence number.
As the \emph{independence number} of a graph is defined as the maximum size of an independent set, the \emph{bipartite independence number} of a bipartite graph is defined as the maximum size of a \emph{balanced} independent set divided by $2$, i.e., the maximum integer $\beta$ such that the bipartite graph admits an independent set of size $2\beta$ taking exactly $\beta$ vertices from each color class.
This result can be regarded as an extension of the solution for complete bipartite graphs, which always have the bipartite independence number $0$.

\begin{theorem}[El Maalouly et al.~\cite{ESW2023}]\label{thm:ESW2023}
The exact matching problem in bipartite graphs can be solved by a deterministic FPT algorithm parameterized by the bipartite independence number.
\end{theorem}

Their result consists of two ingredients.
One is an FPT reduction of EM to a relaxed problem, the \emph{bounded correct parity matching problem (BCPM)} stated as follows, parameterized by the independence number in general and by the bipartite independence number for bipartite graphs.

\begin{problem}[Bounded Correct Parity Matching (BCPM)]
\mbox{ }
\begin{description}
\setlength{\itemsep}{0mm}
\item[Input:] A 0/1-weighted graph $G$ and an integer $k$.
\item[Task:] Determine whether $G$ has a perfect matching of weight $k'$ such that $k' \leq k$ and $k' \equiv k \pmod{2}$.
\end{description}
\end{problem}

\begin{theorem}[El Maalouly et al.~\cite{ESW2023}]\label{thm:ESW2023_reduction}
The exact matching problem can be reduced to the bounded correct parity matching problem by a deterministic FPT algorithm parameterized by the independence number in general and by the bipartite independence number for bipartite graphs, where the input graph does not change.
\end{theorem}

The other is a deterministic polynomial-time algorithm for BCPM in bipartite graphs.

\begin{theorem}[El Maalouly et al.~\cite{ESW2023}]\label{thm:ESW2023_BCPM}
There exists a deterministic polynomial-time algorithm for the bounded correct parity matching problem in bipartite graphs.
\end{theorem}

Their algorithm is based on a standard dynamic programming approach for an equivalent problem.
It seems difficult to be generalized to general graphs (due to existence of so-called \emph{blossoms}), and it remains open whether BCPM can be deterministically solved in polynomial time or not.

In this paper, we try to fill the gap between general graphs and bipartite graphs by considering an \emph{odd cycle transversal}, which is a vertex subset intersecting all the odd cycles in the graph, or equivalently, whose removal makes the graph bipartite.
Our results are stated as follows.

\begin{theorem}\label{thm:FPT}
\mbox{ }
\begin{itemize}
\setlength{\itemsep}{0mm}
\item[(1)] The bounded correct parity matching problem can be solved by a deterministic FPT algorithm parameterized by the minimum size of an odd cycle transversal.
\item[(2)] The exact matching problem can be solved by a deterministic FPT algorithm parameterized by the independence number plus the minimum size of an odd cycle transversal.
The independence number can be sharpened as the bipartite independence number after removing a minimum odd cycle transversal.
\end{itemize}
\end{theorem}

Since the empty set is an odd cycle transversal for any bipartite graph, Theorem~\ref{thm:FPT} extends Theorems~\ref{thm:ESW2023_BCPM} and \ref{thm:ESW2023}.

\subsection{On Correct Parity Matching Problem}
We also consider a further relaxed problem, the \emph{correct parity matching problem (CPM)} stated as follows.

\begin{problem}[Correct Parity Matching (CPM)]
\mbox{ }
\begin{description}
\setlength{\itemsep}{0mm}
\item[Input:] A 0/1-weighted graph $G$ and an integer $k$.
\item[Task:] Determine whether $G$ has a perfect matching of weight $k'$ such that $k' \equiv k \pmod{2}$.
\end{description}
\end{problem}

El Maalouly, Steiner, and Wulf~\cite{ESW2023} also proposed a deterministic polynomial-time algorithm for this problem.
Their algorithm utilizes a linear algebraic trick with the aid of Lov\'{a}sz' algorithm~\cite{Lovasz1987} for finding a basis of the linear subspace spanned by perfect matchings.
It is elegant but heavily depends on the fact that we are only interested in the parity of the weight, and it seems difficult to obtain from it a promising idea to tackle EM.

Let us consider a ``purely graphic'' approach to CPM, which may give some hopeful idea for EM.
We first find a perfect matching $M$.
If the weight of $M$ has the same parity as $k$, we are done.
Otherwise, we solve the \emph{odd alternating cycle problem (OAC)} stated as follows, where an \emph{$M$-alternating cycle} is a simple cycle that alternates between edges in $M$ and not in $M$.

\begin{problem}[Odd Alternating Cycle (OAC)]
\mbox{ }
\begin{description}
\setlength{\itemsep}{0mm}
\item[Input:] A 0/1-weighted graph $G$ and a perfect matching $M$ in $G$.
\item[Task:] Determine whether $G$ has an $M$-alternating cycle of odd weight.
\end{description}
\end{problem}

It is not difficult to observe that the answers of the original CPM instance and of the corresponding OAC instance are the same.
If there exists an $M$-alternating cycle $C$ of odd weight, then the symmetric difference $M \triangle C = (M \setminus C) \cup (C \setminus M)$ is a desired perfect matching (with its parity different from $M$).
Conversely, if there exists a perfect matching $M'$ in $G$ with its parity different from $M$, then the symmetric difference $M \triangle M'$ forms disjoint $M$-alternating cycles having an odd weight in total, which must contain at least one $M$-alternating cycle of odd weight.

Thus, CPM and OAC are polynomial-time equivalent.
One natural way to solve OAC is testing for each edge $e \in M$ whether there exists such a cycle through $e$.
We call this subproblem the \emph{odd alternating cycle through an edge problem (OACe)}, stated as follows.

\begin{problem}[Odd Alternating Cycle through an Edge (OACe)]
\mbox{ }
\begin{description}
\setlength{\itemsep}{0mm}
\item[Input:] A 0/1-weighted graph $G$, a perfect matching $M$ in $G$, and a matching edge $e \in M$.
\item[Task:] Determine whether $G$ has an $M$-alternating cycle of odd weight and through $e$.
\end{description}
\end{problem}

Unfortunately, this problem turns out NP-hard. 
This result implies that, in order to solve CPM, we should search an odd alternating cycle not locally but globally.

\begin{theorem}\label{thm:NP-hard}
The odd alternating cycle through an edge problem is NP-hard even if the input graph is bipartite and contains exactly one (matching) edge of weight $1$.
\end{theorem}

The proof technique is based on a recent result of Schlotter and Seb\H{o}~\cite{SS2022} on the NP-hardness of finding a shortest odd path between two specified vertices in an edge-weighted undirected graph with no negative cycle.

\subsection{Organization}
The rest of this paper is organized as follows.
In Section~\ref{sec:preliminaries}, we describe necessary definitions.
In Section~\ref{sec:FPT}, we prove Theorem~\ref{thm:FPT}, and discuss heuristic speeding-up of our FPT algorithm.
In Section~\ref{sec:NP-hard}, we prove Theorem~\ref{thm:NP-hard}, and discuss further related problems.

\section{Preliminaries}\label{sec:preliminaries}
For basic concepts, terms, and notation on graphs and algorithms, see, e.g., \cite{KV2018, Schrijver2003}.

Let $G = (V, E)$ be a graph with vertex set $V$ and edge set $E$, which is not necessarily simple.
A \emph{(simple) path} $P$ in $G$ is a connected subgraph $(V(P), E(P))$ of $G$, defined by an alternating sequence of vertices and edges $(v_0, e_1, v_1, e_2, v_2, \dots, v_{\ell-1}, e_\ell, v_\ell)$ such that all vertices $v_0, v_1, \dots, v_\ell$ are distinct and $e_i = \{v_{i-1}, v_i\} \in E$ for each $i = 1, 2, \dots, \ell$, where $V(P) = \{v_0, v_1, \dots, v_\ell\}$ and $E(P) = \{e_1, e_2, \dots, e_\ell\}$.
A \emph{(simple) cycle} $C$ in $G$ is a connected subgraph $(V(C), E(C))$ of $G$, defined by an alternating sequence of vertices $(v_0, e_1, v_1, e_2, v_2, \dots, v_{\ell-1}, e_\ell, v_\ell, e_0, v_0)$ such that $P = (v_0, e_1, v_1, e_2, v_2, \dots, v_{\ell-1}, e_\ell, v_\ell)$ is a path and $e_0 = \{v_\ell, v_0\} \in E \setminus E(P)$, where $V(C) = V(P)$ and $E(C) = E(P) \cup \{e_0\}$.
We also deal with paths and cycles just as edge subsets or as vertex sequences (if there is no confusion).
For directed graphs, we define them analogously by replacing undirected edges $\{u, v\}$ between $u$ and $v$ with directed edges $(u, v)$ from $u$ to $v$.

A \emph{matching} in $G$ is a subset of edges no two of which share a vertex.
A matching is said to be \emph{perfect} if it covers all vertices.
A path or cycle is said to be \emph{$M$-alternating} if it alternates between edges in $M$ and not in $M$.

An \emph{independent set} in $G$ is a subset of vertices no two of which are adjacent in $G$.
The maximum size of an independent set in a graph is called the \emph{independence number}.
The complement of an independent set is called a \emph{vertex cover}, i.e., for a vertex cover $X$ in $G$, any edge in $G$ is incident to some vertex in $X$.

A graph is said to be \emph{bipartite} if there exists a (possibly trivial) bipartition $(A, B)$ of its vertex set such that $A$ and $B$ are both independent sets.
When we fix such a bipartition, $A$ and $B$ are called the \emph{color classes} of the bipartite graph.
Then, a \emph{balanced independent set} in the bipartite graph is an independent set $X$ such that $|X \cap A| = |X \cap B|$.
The maximum size of a balanced independent set in a bipartite graph divided by $2$ (i.e., counting the vertices in either color class) is called the \emph{bipartite independence number}.
An \emph{odd cycle transversal} in $G$ is a vertex subset $X$ such that $G - X$ is bipartite, where $G - X$ denotes the subgraph of $G$ obtained by removing the vertices in $X$ with its incident edges.

\section{An FPT Algorithm for Exact Matching Problem}\label{sec:FPT}
\subsection{Proof of Theorem~\ref{thm:FPT}}
In this section, we prove Theorem~\ref{thm:FPT}.

Let $G = (V, E)$ be a 0/1-weighted graph, $k$ be an integer, and $X \subseteq V$ be an odd cycle transversal of $G$.
Let $(A, B)$ be a pair of the two color classes of $G - X$ (if it is not unique, fix an arbitrary one).
For each subset $Y \subseteq X$, we define a bipartite graph $G_Y = (A \cup Y, B \cup (X \setminus Y); E_Y)$ as follows, which is a subgraph of $G$:
\[E_Y \coloneq \{\, \{u, v\} \in E \mid u \in A \cup Y,\ v \in B \cup (X \setminus Y) \,\}.\]
We then have the following lemma.

\begin{lemma}\label{lem:FPT}
$(G, k)$ is a yes-instance of BCPM (or EM) if and only if $(G_Y, k)$ is a yes-instance of BCPM (or EM, respectively) for some $Y \subseteq X$.
\end{lemma}

\begin{proof}
Since each $G_Y$ is a subgraph of $G$, if $G_Y$ has a desired perfect matching, then so does $G$.

Conversely, we show that, for any perfect matching $M$ in $G$, there exists a subset $Y \subseteq X$ such that $G_Y$ has the same perfect matching $M$.
Since $A$ and $B$ are both independent sets in $G$ (as $G - X$ is bipartite), each edge $e = \{u, v\} \in M$ satisfies $|\{u, v\} \cap A| \leq 1$ and $|\{u, v\} \cap B| \leq 1$.

\begin{itemize}
\setlength{\itemsep}{0mm}
\item
If $|\{u, v\} \cap A| = 1$ and $|\{u, v\} \cap B| = 1$, then $e$ always exists in $G_Y$ by definition.

\item
Suppose that $|\{u, v\} \cap A| = 1$ and $|\{u, v\} \cap B| = 0$.
By symmetry, we assume $u \in A$.
In this case, $e \in E_Y$ if and only if $v \in X \setminus Y$.

\item
Suppose that $|\{u, v\} \cap A| = 0$ and $|\{u, v\} \cap B| = 1$.
By symmetry, we assume $u \in B$.
In this case, $e \in E_Y$ if and only if $v \in Y$.

\item
Suppose that $|\{u, v\} \cap A| = 0$ and $|\{u, v\} \cap B| = 0$.
In this case, $e \in E_Y$ if and only if $|\{u, v\} \cap Y| = 1$.
\end{itemize}

Since $M$ is a matching, all the end vertices of the edges in $M$ are distinct.
Therefore, we can consider the above conditions independently for each edge in $M$, and then some $Y \subseteq X$ satisfies the conditions of $e \in E_Y$ for all the edges $e \in M$.
Thus, we are done.
\end{proof}

By Lemma~\ref{lem:FPT}, we can solve BCPM by finding a minimum odd cycle transversal $X$ and by solving BCPM in the bipartite graphs $G_Y$ for all subsets $Y \subseteq X$, i.e., $2^{|X|}$ times.
Since a minimum odd cycle transversal can be found by deterministic FPT algorithms parameterized by its size (initiated by Reed, Smith, and Vetta~\cite{RSV2004}), combining with Theorem~\ref{thm:ESW2023_BCPM}, we obtain Theorem~\ref{thm:FPT}(1).

The first statement of Theorem~\ref{thm:FPT}(2) is an immediate consequence of Theorem~\ref{thm:ESW2023_reduction} and Theorem~\ref{thm:FPT}(1) applied in this order.
To see the second statement of Theorem~\ref{thm:FPT}(2), let us swap the order, i.e., we first find a minimum odd cycle transversal $X$ and apply Lemma~\ref{lem:FPT}, and then apply Theorem~\ref{thm:ESW2023_reduction} to each $G_Y$.
We then obtain an algorithm for EM, which is FPT parameterized by the maximum of the bipartite independence numbers $\beta_Y$ of $G_Y$ plus $|X|$.
As shown below, for every $Y \subseteq X$, the parameter $\beta_Y$ is bounded by the bipartite independence number $\beta$ of $G - X$ plus $|X|$, which concludes the second statement of Theorem~\ref{thm:FPT}(2).

\begin{lemma}\label{lem:BIN}
$\beta_Y \le \beta + |X|$ for any $Y \subseteq X$.
\end{lemma}

\begin{proof}
Let $Z$ be a maximum balanced independent set in $G_Y$.
Then, $Z \cap (A \cup B) = (Z \cap A) \cup (Z \cap B)$ is an independent set in $G - X$, which includes a balanced independent set of size $2 \min \{|Z \cap A|, |Z \cap B|\}$ in $G - X$.
Also, as $Z$ is balanced, we have $|Z \cap A| + |Z \cap Y| = |Z \cap B| + |Z \cap (X \setminus Y)|$.
Thus,
\[2\beta_Y = |Z| = 2 \bigl(\min \{|Z \cap A|, |Z \cap B|\} + \max \{|Z \cap Y|, |Z \cap (X \setminus Y)|\}\bigr) \le {} 2(\beta + |X|). \qedhere\]
\end{proof}

\begin{remark}
The bipartite independence number after removing a minimum odd cycle transversal and the minimum size of an odd cycle transversal are somewhat correlated, but can behave arbitrarily.
For example, by starting with a bipartite graph with small bipartite independence number, choosing a few vertices, and add many edges incident to the chosen vertices, one can construct a graph with both parameters small, for which our FPT algorithm works.
\end{remark}

\subsection{Heuristic Speeding-Up}
In this section, we consider heuristic speeding-up of the algorithm of Theorem~\ref{thm:FPT}.


Let $G = (V, E)$, $k$, $X$, and $(A, B)$ the same as the beginning of Proof of Theorem~\ref{thm:FPT}, and let $n = |V|$.
We assume that $G$ has a perfect matching, since otherwise $(G, k)$ is clearly a no-instance.
For $Y \subseteq X$, if $|A| + |Y| \neq \frac{n}{2}$, then the bipartite graph $G_Y$ has no perfect matching, and hence $(G_Y, k)$ is a no-instance.
Therefore, we can solve BCPM/EM by solving BCPM/EM in the bipartite graphs $G_Y$ for all $Y \subseteq X$ with $|Y| = \frac{n}{2} - |A|$, i.e., $\binom{|X|}{\frac{n}{2}-|A|}$ times rather than $2^{|X|}$ times.
To make $\binom{|X|}{\frac{n}{2}-|A|}$ small, we consider the following problem, called the \emph{unbalanced bipartization problem (UB)}.

\begin{problem}[Unbalanced Bipartization (UB)]
\mbox{ }
\begin{description}
\setlength{\itemsep}{0mm}
\item[Input:] A graph $G$ having a perfect matching, and two integers $k$ and $l$.
\item[Task:] Determine whether $G$ has an odd cycle transversal $X\subseteq V$ satisfies the following conditions:
\begin{itemize}
\setlength{\itemsep}{1mm}
    \item $|X|\le k$.
    \item $G - X$ has two color classes $(A,B)$ such that $|A|\ge \frac{n}{2}-l$.
\end{itemize}
\end{description}
\end{problem}

\begin{lemma}\label{lem:ub}
The unbalanced bipartization problem is NP-hard.
\end{lemma}

\begin{proof}
We reduce to UB the \emph{odd cycle transversal problem (OCT)} stated as follows, which is NP-hard and admits FPT algorithms parameterized by $k$ (initiated by \cite{RSV2004}).

\begin{problem}[Odd Cycle Transversal (OCT)]
\mbox{ }
\begin{description}
\setlength{\itemsep}{0mm}
\item[Input:] A graph $G$ and an integer $k$.
\item[Task:] Determine whether $G$ has an odd cycle transversal $X\subseteq V$ of size at most $k$.
\end{description}
\end{problem}

If $G$ has a perfect matching, then the OCT instance $(G,k)$ reduces to an UB instance $(G,k,\frac{n}{2})$.
Thus, it suffices to prove that OCT is NP-hard even if the input graph has a perfect matching.

Suppose that we are given an OCT instance $(G = (V, E), k)$, where $G$ does not necessarily have a perfect matching.
We define a graph $G'=(V',E')$ as follows:
\begin{align*}
    V'&\coloneq V\cup \{\, v':v\in V \,\},\\
    E'&\coloneq E\cup \{\, \{v, v'\} : v\in V \,\}.
\end{align*}

Clearly, $G'$ has a perfect matching consisting of the additional edges.
Also, since the degree of each additional vertex is $1$, no cycle intersects them, and hence the two OCT instances $(G, k)$ and $(G', k)$ are equivalent.
Thus, we are done.
\end{proof}

If $l \geq \frac{k}{2}$, then the UB instance $(G, k, l)$ is equivalent to the OCT instance $(G,k)$ ($|X|\le k$ implies that $|A|\ge \frac{n-k}{2}$ or $|B|\ge \frac{n-k}{2}$).
Thus, we assume $l < \frac{k}{2}$.
A na\"{i}ve brute-force algorithm for UB takes $n^{k+O(1)}$ time.
We can design an XP algorithm with a better computational time bound as follows.

\begin{lemma}
The unbalanced bipartization problem can be solved in $2.3146^k n^{l+O(1)}$ time.
\end{lemma}

\begin{proof}
This proof is based on an FPT algorithm for OCT parameterized by $k$ through a reduction to the \emph{unweighted vertex cover problem (UVC)} (cf.~\cite{Cygan2015}).

\begin{problem}[Unweighted Vertex Cover (UVC)]
\mbox{ }
\begin{description}
\setlength{\itemsep}{0mm}
\item[Input:] A graph $G$ and an integer $k$.
\item[Task:] Determine whether $G$ has a vertex cover of size at most $k$.
\end{description}
\end{problem}

The \emph{weighted vertex cover problem (WVC)} is a weighted variant of UVC.

\begin{problem}[Weighted Vertex Cover (WVC)]
\mbox{ }
\begin{description}
\setlength{\itemsep}{0mm}
\item[Input:] A graph $G$, a vertex-weight function $w \colon V \to \mathbb{Z}_{\ge 0}$, and an integer $k$.
\item[Task:] Determine whether $G$ has a vertex cover of weight at most $k$.
\end{description}
\end{problem}

The proof is sketched as follows.
We first reduce UB to WVC, and then reduce WVC to UVC by the result of Niedermeier and Rossmonith~\cite{NR2003} (Lemma~\ref{non-weighted}).
Finally, we solve UVC by an FPT algorithm proposed by Lokshtanov, Narayanaswamy, Raman, Ramanujan, and Saurabh~\cite{LNRRS2014}
 (Lemma~\ref{AGVC}).

\begin{lemma}[Niedermeier and Rossmonith~\cite{NR2003}]\label{non-weighted}
Let $(G = (V, E), w, k)$ be a WVC instance.
For each vertex $v \in V$, let $C(v)$ be the set of $w(v)$ copies of $v$.
We define a graph $\tilde{G} = (\tilde{V}, \tilde{E})$ as follows:
\begin{align*}
\tilde{V} &\coloneq \bigcup_{v \in V} C(v),\\
\tilde{E} &\coloneq \{\, \{\tilde{u}, \tilde{v}\} : \{u, v\} \in E,\ \tilde{u} \in C(u),\ \tilde{v} \in C(v) \,\}.
\end{align*}
Then, the WVC instance $(G, w, k)$ and the UVC instance $(\tilde{G}, k)$ are equivalent.
\end{lemma}

\begin{lemma}[Lokshtanov et al.~\cite{LNRRS2014}]\label{AGVC}
There exists a deterministic algorithm for UVC running in $2.3146^{k - \mu} n^{O(1)}$ time, where $\mu$ is the maximum size of a matching in the graph.
\end{lemma}

Suppose that we are given a UB instance $G = (V, E)$, $k$, and $l$.
Let $M$ be a perfect matching in $G$.
For each subset $F\subseteq M$, we construct a corresponding WVC instance $(G^F, w_F, k_F)$ as follows.

A graph $G^F = (V^F, E^F)$ is defined as follows, where $V(F)$ denotes the set of end vertices of edges in $F$:
\begin{align*}
V^F_1 &\coloneq \{\, v_1 : v \in V \setminus V(F) \,\}, \\
V^F_2 &\coloneq \{\, v_2 : v \in V \,\}, \\
V^F &\coloneq V^F_1 \cup V^F_2, \\
E^F &\coloneq \{\, \{v_1, v_2\} : v \in V \setminus V(F) \,\} \cup \{\, \{u_i, v_i\} : \{u, v\} \in E,\ u_i, v_i \in V^F_i,\ i \in \{1,2\} \,\}.
\end{align*}
A vertex-weight function $w_F \colon V^F \to \mathbb{Z}_{\ge 0}$ is defined as
\[w_F(v) \coloneq
\begin{cases}
    k + 1 & (v \in V^F_1),\\
    1 & (v \in V^F_2).\\
\end{cases}\]
Let us define
\[k_F = \left(\frac{n}{2} - |F|\right)(k + 2) + k.\]
We then have the following lemma.

\begin{lemma}\label{lem:UB_WVC}
$(G,k,l)$ is a yes-instance of UB if and only if $(G^F, w_F, k_F)$ is a yes-instance of WVC for some $F \subseteq M$ with $|F| \le l$.
\end{lemma}

\begin{proof}
Suppose that $G$ has a desired odd cycle transversal $X\subseteq V$.
Then, $G - X$ has two color classes $(A, B)$ such that $|A|\ge \frac{n}{2} - l$.
Let $F = \{\, e : e\in M,\ e \cap A = \emptyset \,\}$, $A' = \{\, v_1 : v\in A \,\} \subseteq V_1^F$, and $B'=\{\, v_2:v\in B \,\} \subseteq V_2^F$.
As $A'\cup B'$ is an independent set in $G^F$, its complement $V^F \setminus (A'\cup B')$ is a vertex cover in $G^F$.
Since $|A'| = |A| = |M| - |F| = \frac{n}{2} - |F|$ and $|B'| = |B| = n - |X| - |A| \geq \frac{n}{2} - k + |F|$, its weight is
\begin{align*}
    w_F(V^F)-w_F(A'\cup B') &= (|V_1^F| - |A'|)(k + 1) + |V_2^F| - |B'|\\
    &\le \left((n-2|F|) - \left(\frac{n}{2}-|F|\right)\right)(k+1) + n -\left(\frac{n}{2}-k+|F|\right)\\
    &= k_F.
\end{align*}
Thus, $V^F \setminus (A'\cup B')$ is a desired vertex cover.
As $|A| \ge \frac{n}{2} - l$, we have $|F| = |M| - |A| \le l$.

Conversely, suppose that $G^F$ has a desired vertex cover $X$ for some $F \subseteq M$ with $|F| \le l$.
We first show that
\begin{align}
\label{equation:size_V1}
|X\cap V_1^F|=\frac{n}{2}-|F|.
\end{align}

Since $V_1^F$ has a perfect matching corresponding to $M \setminus F$, we have $|X\cap V_1^F|\ge\frac{n}{2}-|F|$.
To derive a contradiction, suppose that $|X\cap V_1^F| > \frac{n}{2}-|F|$.
Since $V_2^F$ also has a perfect matching corresponding to $M$, we have $|X\cap V_2^F|\ge \frac{n}{2}$.
Thus, we obtain
\begin{align*}
    w_F(X) &= w_F(X\cap V_1^F) + w_F(X\cap V_2^F)\\
    &\ge \left(\frac{n}{2} - |F| + 1\right)(k + 1) + \frac{n}{2}\\
    &= k_F + |F| + 1,
\end{align*}
which contradicts $w_F(X) \le k_F$.
Thus, we have \eqref{equation:size_V1}.

From \eqref{equation:size_V1} and $w_F(X)\le k_F$, we obtain
\begin{align}
    |X\cap V_2^F| \le k_F-w_F(X\cap V_1^F) = \frac{n}{2}-|F|+k. \label{equation:size_V2}
\end{align}

Let $Z = V^F \setminus X$, $A' = Z \cap V^F_1$, and $B' = Z \cap V^F_2$.
Let $A$ and $B$ denote the vertex subsets in $G$ corresponding to $A'$ and $B'$, respectively.
We show that $V \setminus (A\cup B)$ is a desired odd cycle transversal in the UB instance $(G, k, l)$.

Since $X$ is a vertex cover in $G^F$, its complement $Z$ is an independent set in $G^F$.
Since $G^F$ has an edge $\{v_1, v_2\}$ for every $v \in V \setminus V(F)$, we have $A \cap B = \emptyset$.
In addition, both $A$ and $B$ are independent sets in $G$, and hence $V \setminus (A\cup B)$ is an odd cycle transversal in $G$.
From \eqref{equation:size_V1} and $|F|\le l$, we obtain
\begin{align*}
|A| = |V^F_1|-|X\cap V^F_1| = \frac{n}{2}-|F| \ge \frac{n}{2} - l. 
\end{align*}
Combined with \eqref{equation:size_V2}, we have
\begin{align*}
|V \setminus (A\cup B)| &= |V|-|A|-|B|\\
&= |V|-|A|-(|V^F_2|-|X\cap V^F_2|)\\
&\le n-\left(\frac{n}{2}-|F|\right)-\left(n-\left(\frac{n}{2}-|F|+k\right)\right)\\
&= k. 
\end{align*}
Thus, $V \setminus (A \cup B)$ is a desired odd cycle transversal.
\end{proof}

By Lemma~\ref{lem:UB_WVC}, we can solve a UB instance $(G, k, l)$ by solving the WVC instances $(G_F,w,k_F)$ for all subsets $F\subseteq M$ with $|F| \le l$, i.e., $O(n^l)$ times.
By Lemma~\ref{non-weighted}, we can define the corresponding UVC instance $(\tilde{G}^F = (\tilde{V}^F, \tilde{E}^F), k_F)$ as follows:
\begin{align*}
\tilde{V}^F &\coloneq \bigcup_{v \in V^F} C(v),\\
\tilde{E}^F &\coloneq \{\, \{\tilde{u}, \tilde{v}\} : \{u, v\} \in E^F,\ \tilde{u} \in C(u),\ \tilde{v} \in C(v) \,\}.
\end{align*}

As $G^F$ has a perfect matching $M^F = \{\, \{u_1, v_1\} : \{u, v\} \in M \setminus F \,\}\cup \{\, \{u_2, v_2\} : \{u, v\} \in M \,\}$, $\tilde{G}^F$ has a corresponding perfect matching $\tilde{M}^F$. 
By Lemma~\ref{AGVC}, we can solve each UVC instance $(G'_F,k_F)$ in $2.3146^{k_F-|\tilde{M}^F|} n^{O(1)}$ time, where
\begin{align*}
k_F - |\tilde{M}^F| &=\left(\frac{n}{2}-|F|\right)(k+2)+k-\left(\left(\frac{n}{2}-|F|\right)(k+1)+\frac{n}{2}\right)\\
&= k - |F|\\
&\le k.
\end{align*}
Thus, we are done.
\end{proof}

\section{NP-hardness of Related Problems}\label{sec:NP-hard}
\subsection{Proof of Theorem~\ref{thm:NP-hard}}
In this section, we prove Theorem~\ref{thm:NP-hard}.
The proof is similar to that of \cite[Theorem~4.3]{SS2022}.
We reduce to OACe the \emph{back and forth paths problem (BFP)} stated as follows, which is known to be NP-hard~\cite{FHW1980}.

\begin{problem}[Back and Forth Paths (BFP)]
\mbox{ }
\begin{description}
\setlength{\itemsep}{0mm}
\item[Input:] A directed graph $G$ and two vertices $s$ and $t$.
\item[Task:] Determine whether $G$ has a simple cycle that contains both $s$ and $t$.
\end{description}
\end{problem}

Suppose that we are given a BFP instance $G = (V, E)$ and $s, t \in V$.
We construct the corresponding OACe instance, i.e., a 0/1-weighted graph $\hat{G} = (\hat{V}, \hat{E})$, a perfect matching $\hat{M}$, and $\hat{e} \in \hat{M}$ as follows.

We split each $v \in V$ into two vertices, \emph{in-copy} $v^-$ and \emph{out-copy} $v^+$.
For each directed edge $(u, v) \in E$, we create an edge $\{u^+, v^-\}$ of weight $0$.
For each vertex $v \in V$, we create an edge $\{v^-, v^+\}$, whose weight is $1$ if $v = t$ and $0$ otherwise.
Define $\hat{M}$ as those edges $\{v^-, v^+\} \ (v \in V)$, and let $\hat{e}$ be $\{s^-, s^+\}$.
To sum up, we define
\begin{align*}
\hat{V} &\coloneq \{\, v^- : v\in V \,\} \cup \{\, v^+ : v\in V \,\},\\
\hat{E} &\coloneq \{\, \{u^+, v^-\} : (u, v) \in E \,\} \cup \{\, \{v^-, v^+\} : v\in V \,\},\\
\hat{M} &\coloneq \{\, \{v^-, v^+\} : v\in V \,\},\\
\hat{e} &\coloneq \{s^-, s^+\},
\end{align*}
where only $\{t^-, t^+\}$ has weight $1$ .
Note that this construction satisfies the additional conditions (bipartiteness and uniqueness of an edge of weight 1) in Theorem~\ref{thm:NP-hard}.

We show that there exists a simple cycle containing $s$ and $t$ in $G$ if and only if there exists an $\hat{M}$-alternating cycle of odd weight through $\hat{e} = \{s^-, s^+\}$ in $\hat{G}$.

Suppose that there exists a simple cycle $C = (s, v_1, v_2, \dots, t, \dots, v_\ell, s)$ containing $s$ and $t$ in $G$.
Then, there exists a corresponding simple cycle $\hat{C} = (s^-, s^+, v_1^-, v_1^+, v_2^-, v_2^+, \dots, t^-, t^+, \dots, v_\ell^-, v_\ell^+, s^-)$ in $\hat{G}$, which is clearly $\hat{M}$-alternating, of weight $1$ (due to the edge $\{t^-, t^+\}$), and through $\hat{e} = \{s^-, s^+\}$.

Conversely, suppose that there exists an $\hat{M}$-alternating cycle $\hat{C}$ of odd weight through $\hat{e} = \{s^-, s^+\}$ in $\hat{G}$.
Since $\{t^-, t^+\}$ is the only edge having weight $1$, it must be traversed by $\hat{C}$.
In addition, since the graph $\hat{G}$ is bipartite and the two color classes correspond to the signs, $\hat{C}$ alternately intersects $u^-$ and $v^+$ for some vertices $u, v \in V$.
That is, $\hat{C}$ is a form of 
\[(s^-, s^+, v_1^-, v_1^+, v_2^-, v_2^+, \dots, t^-, t^+, \dots, v_\ell^-, v_\ell^+, s^-),\]
and hence there exists a corresponding simple cycle $C = (s, v_1, v_2, \dots, t, \dots, v_\ell, s)$ in $G$.
Thus, we are done.

\subsection{Further Related Problems}
In this section, we discuss several more problems related to OACe.

For a matching $M$, an \emph{$M$-augmenting path} is a simple path between two unmatched vertices that alternates between edges in $M$ and not in $M$.
This is an elementary but crucial structure in the maximum matching problem.
In particular, a matching $M$ is not of maximum size if and only if there exists an $M$-augmenting path.

OACe reduces to the \emph{odd augmenting path problem (OAP)} stated as follows; thus, OAP is also NP-hard.

\begin{problem}[Odd Augmenting Path (OAP)]
\mbox{ }
\begin{description}
\setlength{\itemsep}{0mm}
\item[Input:] A 0/1-weighted graph $G$ and a matching $M$ in $G$.
\item[Task:] Determine whether $G$ has an $M$-augmenting path of odd weight.
\end{description}
\end{problem}

\begin{lemma}\label{lem:OAP}
OACe reduces to OAP.
\end{lemma}

\begin{proof}
Suppose that we are given an OACe instance, i.e., a 0/1-weighted graph $G = (V, E)$, a perfect matching $M$, and a matching edge $e = \{u, v\} \in M$.
Consider a graph $G' = (V, E \setminus \{e\})$ and a matching $M' = M \setminus \{e\}$ in $G'$.
In addition, if the weight of $e$ is $1$, then we switch the weight (change it from $0$ to $1$ and from $1$ to $0$) of every edge incident to $u$ in $G'$.
Then, $G$ has an $M$-alternating cycle of odd weight and through $e$ if and only if $G'$ has an $M'$-augmenting path of odd weight.
The former is indeed transformed into the latter just by removing $e$ (note that both traverse exactly one edge incident to $u$ in $G'$, which preserves the parity of their weights).
Also, since $u$ and $v$ are the only unmatched vertices with respect to $M'$, the latter is transformed into the former by adding $e$.
Thus, we are done.
\end{proof}

The proof and Theorem~\ref{thm:NP-hard} imply the following.

\begin{corollary}\label{cor:OAP}
The odd augmenting path problem is NP-hard even if the input graph is bipartite and contains exactly one (matching) edge of weight $1$ and the input matching is of size $\frac{n}{2} - 1$, where $n$ is the number of vertices.
\end{corollary}

OAP with this additional conditions reduces to the \emph{disjoint augmenting path problem (DAP)} stated as follows; thus, DAP is also NP-hard.

\begin{problem}[Disjoint Augmenting Path (DAP)]
\mbox{ }
\begin{description}
\setlength{\itemsep}{0mm}
\item[Input:] A graph $G$, a matching $M$ in $G$, and unmatched vertices $s_1, s_2, t_1, t_2$.
\item[Task:] Determine whether $G$ has disjoint $M$-augmenting paths between $\{s_1, s_2\}$ and $\{t_1, t_2\}$.
\end{description}
\end{problem}

The reduction is easy.
Let $s_1, s_2$ be the original unmatched vertices, let $t_1, t_2$ be the end vertices of the unique matching edge $e$ of weight $1$, and remove $e$.

\begin{corollary}
The disjoint augmenting path problem is NP-hard even if the input graph is bipartite.
\end{corollary}

In contrast, a relaxed problem, the \emph{free disjoint augmenting path problem (FDAP)} stated as follows, is polynomial-time solvable.

\begin{problem}[Free Disjoint Augmenting Path (FDAP)]
\mbox{ }
\begin{description}
\setlength{\itemsep}{0mm}
\item[Input:] A graph $G$, a matching $M$ in $G$, and unmatched vertices $s_1, s_2, t_1, t_2$.
\item[Task:] Determine whether $G$ has disjoint $M$-augmenting paths with end vertices in $\{s_1, s_2, t_1, t_2\}$.
\end{description}
\end{problem}

\begin{lemma}
FDAP can be solved by a deterministic polynomial-time algorithm.
\end{lemma}

\begin{proof}
Let $G'$ be a subgraph of $G$ obtained by removing the unmatched vertices other than $s_1, s_2, t_1, t_2$.
If $G'$ has a perfect matching $M'$, then the symmetric difference $M \triangle M'$ contains desired disjoint $M$-augmenting paths.
Conversely, if $G$ contains desired disjoint $M$-augmenting paths, so does $G'$, which implies that $G'$ has a perfect matching.
Thus, it suffices to solve the maximum matching problem in $G'$, for which a deterministic polynomial-time algorithm exists.
\end{proof}

\section*{Acknowledgments}
The authors are grateful to anonymous reviewers for their helpful comments.
The second author was supported by JSPS KAKENHI Grant Numbers 20K19743 and 20H00605.


\end{document}